\pdfoutput=1
\documentclass[12pt, draftclsnofoot, onecolumn]{IEEEtran}
\usepackage{amsmath,amsfonts,amssymb,amsbsy, amsthm,ae,aecompl}
\usepackage{algorithm,algpseudocode}
\usepackage[utf8]{inputenc}
\usepackage[english]{babel}
\usepackage{bm}
\usepackage{color}
\usepackage[noadjust]{cite}
\usepackage{epsfig}
\usepackage{enumerate}
\usepackage{float} 
\usepackage{fancyhdr}
\usepackage[T1]{fontenc}
\usepackage[acronym,toc,shortcuts]{glossaries}
\usepackage{graphicx, caption, subcaption}
\usepackage{hyperref}
\usepackage{lastpage}
\usepackage{listings}
\usepackage{lipsum}
\usepackage{dsfont}
\usepackage{multirow,tabularx}
\usepackage[normalem]{ulem}
\usepackage{tikz,pgfplots}
\usepackage{times}
\usepackage{verbatim}
\usepackage[all]{xy}
\usepackage{mathtools}
\usepackage{tikz}
\usepackage{tikz-qtree,tikz-qtree-compat}
\usepackage{pgfplots}

\usepackage{lipsum}
\usetikzlibrary{calc}
\usetikzlibrary{patterns}

\usepackage{enumitem}
\graphicspath{{figures/}}

\usetikzlibrary{arrows,shapes}
\newcommand{\E}[1]{{\mathbb E}\left[ #1 \right]}

\newtheorem{Theorem}{Theorem}

\hypersetup{
	colorlinks,%
	citecolor=black,%
	filecolor=black,%
	linkcolor=black,%
	urlcolor=black
}

\pgfplotsset{
	grid style = {
		dash pattern = on 0.025mm off 0.95mm on 0.025mm off 0mm, 
		line cap = round,
		black,
		line width = 0.5pt
	},
	tick label style={font=\small},
	label style={font=\small},
	legend style={font=\footnotesize},
}

\pgfplotscreateplotcyclelist{laneas1}{
	cyan!60!black,		solid, 	every mark/.append style={fill=cyan!60!black},mark=x\\%
	cyan!60!black,		solid, 	every mark/.append style={fill=cyan!60!black},mark=+\\%
	cyan!60!black,		solid, 	every mark/.append style={fill=cyan!60!black},mark=o\\%
	red!80!black,       dashed\\%
	cyan!60!black, 		densely dotted,	every mark/.append style={fill=cyan!80!black},mark=diamond*\\%
}

\ifCLASSINFOpdf
\else
\fi
\newtheorem{lem}{Lemma}

\begin{document}

	\title{Energy Efficiency in Cache Enabled Small Cell Networks With Adaptive User Clustering}

	\author{\IEEEauthorblockN{Salah Eddine Hajri and Mohamad Assaad }\\
		\IEEEauthorblockA{Laboratoire des Signaux et Systemes (L2S, CNRS), CentraleSupelec,
			Gif-sur-Yvette, France\\
			Email: \{Salaheddine.hajri,\; Mohamad.Assaad\}@centralesupelec.fr}

	}

	\maketitle
	
	\begin{abstract}
		
		Using a network of cache enabled small cells, traffic during peak hours can be reduced considerably through proactively fetching the content that is most probable to be requested. In this paper, we aim at exploring the impact of proactive caching on an important metric for future generation networks, namely, energy efficiency  (EE).  We argue that, exploiting the correlation in  user content popularity profiles in addition to the spatial repartitions of users with comparable request patterns, can result in considerably improving the achievable energy efficiency of the network.  In this paper, the problem of optimizing EE is decoupled into two related subproblems. The first one addresses the issue of content popularity modeling. While most existing works assume similar popularity  profiles for all users in the network, we consider an alternative caching framework in which,  users are clustered according to their content popularity profiles.  In order to showcase the utility of the proposed clustering scheme, we use a statistical model selection criterion, namely \emph{Akaike information criterion} (AIC).  Using stochastic geometry, we derive a closed-form expression of the achievable EE and we find the optimal active small cell density vector that maximizes it. The second subproblem investigates the impact of exploiting the spatial repartitions of users with comparable request patterns. After considering a snapshot of the network, we formulate a combinatorial optimization problem that enables to optimize content placement such that the used transmission power is minimized.
		Numerical results show that the clustering scheme enable to considerably improve the cache hit probability and consequently the EE compared with an unclustered approach. Simulations also show that the small base station allocation algorithm   results in improving the energy efficiency and hit probability.

	\end{abstract}
	

	% \fi
	\IEEEpeerreviewmaketitle
	
	\section{Introduction}

	The global mobile traffic is expected to increase exponentially in the coming years.  This is mainly due to the wide spread of wireless devices and the emergence of video streaming as one of the main contributors in mobile data traffic. Network densification through the deployment  of small base stations (SBSs) was proposed as a  mean  to offload an important amount of traffic from the macro base stations. However, this requires the deployment of a costly infrastructure and a considerable increase in the backhaul link capacity.	Recently, information centric networks are emerging as an efficient technology to offload traffic and reduce the strains on the backhaul. In fact, a substantial part of the mobile traffic is due to several duplicate requests of the same popular contents. Consequently, proactively caching these files in the edge of the network results in enhancing user experience while reducing the needed backhaul link capacity. In addition to satisfying the increasing demand, sustainable development is also a major requirement for 5G networks. In fact, the increasing concern regarding green gazes emission made EE a major key performance indicators for future networks.  Thanks to the improvement of memory devices, proactive caching offers a very practical and energy efficient alternative to network densification.
	
	The idea of caching popular content on the edge of the network is gaining momentum as one of the most promising enablers of future generation networks \cite{informationcentric},  \cite{MATHA}. Based on the observation that the backhaul link is becoming a bottleneck, especially in dense networks, the idea of exploiting memory devices as a substitute to more backhaul capacity proves to be very tempting. Many existing works investigated the fundamental trade-offs in cache enabled networks.  In \cite{algo}, a joint routing and caching problem in small cell networks was considered, taking into account both the constrained storage and transmission bandwidth capacities of the SBSs.  The authors used approximation algorithms with performance guarantees in order to derive a solution that maximizes the content requests that are satisfied by the SBSs. In \cite{limits}, an information theoretic formulation of the caching problem was considered. The authors proposed coded schemes enabling a considerable improvement in peak rate compared to previously known schemes. In \cite{edge}, with a limited backhaul  capacity and proactive caching exploiting context awareness and social networks, the authors showed that the backhaul traffic load can be substantially  reduced.   A hierarchical caching system with two layers of caches was considered in \cite{code}. The authors proposed a coded caching scheme that attains the optimal memory-rate trade-off to within a constant gap. The relation between collaboration distance and interference was studied in \cite{dd} for D2D networks. The authors showed that with enough content reuse, non-vanishing throughput per user can be attained, even with limited storage and delay. Clustering users according to their request pattern was investigated in \cite{cont} with the goal of reducing service delay. The authors showed that the clustering scheme outperforms the unclustered and random caching approach.
	
	EE of cache enabled networks is a fundamental subject that is attracting increasing attention   \cite{rev1,rev4}. 
	The impact of proactive caching on EE was investigated in \cite{EED}. The  key factors that impact the EE of cache enabled networks were studied.  The authors showed that EE can be improved by caching at the base stations when, power efficient cache hardware and sufficient cache capacity are used. In \cite{rev1}, a GreenDelivery framework was proposed in cache enabled small cell networks. Using energy harvesting communications, the authors showed that this framework enables to  reduce energy consumption. In \cite{mar}, the  energy consumption  of cache enabled wireless cellular networks was investigated. The authors studied the  conditions under  which  the  area  power  consumption  is  minimized  while  ensuring  a high  coverage  probability.

	In this paper,   we  explore the impact of the content placement strategy on EE. While most previous works  assume  similar  popularity profiles for all users, we consider an alternative caching framework in which,  users are clustered according to their content popularity profiles \cite{conf_version, cont}. This choice is motivated by the existence of very diverse traffic patterns among users. In fact, the requested content depends on the user social network and interests that can be very different from one person to the other. Assuming a homogeneous content popularity among users can only result in loosing valuable information. Contrary  to  classical  location based   clustering approaches, we choose  to use content based clustering. This  can be  justified by  the fact that content popularity   change slower than user locations. Owing to the different time scales of content popularity and user location changes, the selected cached files,  which depends on the average  popularity distribution per cluster, can be kept  constant for long periods of time regardless of user positions. 
	
	In order to showcase the pertinence of the proposed popularity based clustering scheme, we use a statistical model selection criterion, namely, \emph{Akaike information criterion}.  AIC enables to measure the truthfulness of a given statistical model.  It also addresses the trade-off between the fitness of a statistical model based on maximum likelihood estimation and its complexity which is given by the number of parameters to be estimated. { AIC enables to adapt user clustering  to  any traffic pattern changes since it can detect modifications in the optimal number of clusters. We find that  content popularity based clustering enables a substantial  gain  in term  of cache hit probability and  EE, even when user positions are not taken into consideration.
		Nevertheless, further improvement is possible by exploiting geographic information. Of course, acquiring  information on user location requires a non negligible processing and signaling overhead. Consequently, this information should be leveraged whenever it is available. In the second part of the paper, we develop an optimization framework to exploit any spatial correlation in traffic patterns. }
	
	In this paper, we choose to  tackle the problem  of  file  placement  once content based clustering is done. The choice is motivated by the different time scales according to  which content popularity and user location evolve.
	In fact, while the correlation in content popularity  between users from the same social  group is constant for long periods of time, their location can  change  due to  mobility. This motivates the need to adapt the cached content placement more often than the selected files in order to simplify the management of the network. 	{ In the case of  low mobility, where users do not change positions too often, it makes sense to adapt the files placement based on location information.  This is the case for users in confined areas (office, university campus....).}
	To this end, we propose an optimization framework that enables to exploit any spatial correlation between users with comparable popularity profiles.  We consider a setting in which  user and SBS  location is known based on a  snapshot of their  Poisson point processes (PPP).
	The proposed combinatorial problem aims at associating each individual SBS with the files of a given cluster. This optimization 
	enables a more energy efficient scheme that exploit the spatial correlation in user traffic pattern in order to reduce the average consumed power. It also showcases another interesting advantage of content based clustering. In fact, the clustering that is done on the users enables also to group the files accordingly. Consequently, the complexity of the resulting optimal file placement  problem is lower since the search space is reduced from the whole file catalog to groups of file of approximately equal total size. This considerably simplifies the management of the caching system compared to  existing work on location based optimization where, the complexity of the formulated problems is proportional  to  the  number of files.	
	\subsection{ Contribution and outcomes}
	{The main contributions of our work are presented as follows:
		\begin{enumerate}
			\item A clustering framework for caching: 
			Given  heterogeneous user profiles, we propose a content popularity based clustering scheme. In order to achieve an efficient user grouping, we use the Akaike Information Criterion. This  allows to effectively estimate the  number of clusters  and the associated average popularity  profiles. 
			\item Optimal active SBSs density: We derive a closed form expression of the achievable
			EE. We then optimize the achievable EE with respect to the density vector of active SBSs. This results in improving the achievable EE even  when user positions are not taken into consideration. This is useful in practical scenarios where acquiring user locations requires  substantial signaling.
			\item Optimizing small base station allocation: 
			When information about user location is available, further improvement  of  EE can be achieved by optimizing the SBSs allocation to the different clusters. We formulate a combinatorial content placement problem that enables to 
			adapt the allocation of cached content based on user location. 
			%The problem aims at caching in each small base station the files that
			%	are most likely to be requested within their  neighborhood. This optimization, consequently, results in improving the achievable EE.
		\end{enumerate}}
		The paper is organized as follows: We describe the system model in Section II.  User clustering  will be investigated in Section III.  EE will then be addressed in section IV. In section V, we present the SBSs allocation algorithm.  Finally, in Section  VI, numerical results are presented. 
		\section{ SYSTEM MODEL AND PRELIMINARIES}	
		\subsection{Network Model}		
		We consider a small cell network deployed over a disc with  radius $R_n$. The SBSs  are spatially distributed  according to a homogeneous Poisson point process  $ \phi_s$ with density $\lambda_{s_{max}}$. In this paper, the available SBSs can be in idle or active modes. The density of active SBSs is given by $\lambda_{s}$ such that $\lambda_{s} \leq \lambda_{s_{max}}$.
		We  consider an orthogonal frequency-division multiple-access (OFDMA) system where,  users served by the same SBS, are scheduled on orthogonal  resources. Consequently, each user will be subject to interference coming  from  users served  by  other SBSs.
		The  users are also distributed  in $\mathbb{R}^2$  according to an independent homogeneous PPP
		$ \phi$ with density $\lambda$ such that,  $\lambda>>\lambda_{s_{max}}$.  The average number of users in the network is then given by $U=\lambda \pi R^2_n$.
		Each user is equipped with a single antenna and is allowed to  communicate with  any SBS within a radius $R$. This restriction enables to control the level  of interference.  We consider that $R$ is defined so that each user is covered with high probability  by more than one SBS.
		We consider that a   packet can be successfully transmitted and decoded if and  only if  SINR  $>\theta$. This means  that, if the SINR is lower than the threshold $\theta$, the link undergoes an outage and the transmission fails.	
		A general power law  pathloss model is used where, the power decay is given by $r^{-\alpha}_{us}$.  $r_{us}$ represents the distance between user $u$ and its serving SBS $s$ and $\alpha >2$ denotes the  pathloss exponent.The wireless channel  from  user $u$ to the SBS $s$ is then given by:	
		\begin{align}\tag{1}
		& \nonumber{g}_{ us } = \sqrt{r^{-\alpha}_{us}} {h}_{ us }, \nonumber
		\end{align}	
		where ${h}_{ us }$ represents the small scale fading coefficient modeled as  Rayleigh fading i.e., $CN\left(0,1 \right)$ distributed random variable. 
		We consider that  the transmit power, used in both uplink and downlink, is defined  according to channel inversion power control \cite{power}. This  is done so that the transmit power compensates the pathloss in order to keep the average signal power  at the  receiver (i.e., the SBS or the user terminal)  equal to a certain constant value $\rho_0$. The transmit power used by user $u$ to communicate with SBS $s$, according to channel inversion power control, is given by: $\rho_{us}= \rho_0 r^{\alpha}_{us}$.
		The channel inversion power control will  ensure a limitation of the interference level since the power received at  any base station from a typical user is upper bounded by $\rho_0 R^{\alpha}$, where $R$ denotes the  maximum communication radius.
		Controlling the level of interference,   in both uplink and  downlink, is a vital  factor that  guarantees {an  EE} gain \cite{EED}.
		\subsection{User scheduling and caching strategy}	
		We  consider  a file catalog $C$ containing $F$  files with different sizes. Each file $i$ has a size of $L_i$ bits.	
		In this paper, we consider that the users  have  heterogeneous file popularity distributions. Each user $u$ is associated with a  popularity vector $P_u=\left[ p_{1u}...p_{Fu}\right] $, where $p_{iu}$ denotes the probability that user $u$  requests file $i$ from the catalog.  We consider that  these probabilities change slowly over time and that they are previously known by the network.  Estimating the popularity distributions can be performed by learning from   previously recorded requests \cite{bandit}. {  In this work, we limit our analysis to the case of perfectly known popularity distributions. The study of the impact of estimation error in popularity distributions is considered in future work.}
		Although users have heterogeneous popularity profiles, we assume that they  can be grouped according to their interest into $N_c$ clusters. This means that the users, forming each cluster, have correlated request patterns.  Meaning that the distance between their content popularity vectors is small.
		Each SBS is equipped with a  caching capacity of $M$ bits. 
		{ Each individual SBS  fills its memory device with the  most popular files from a given cluster. In each cluster,  the most  popular files are  selected based on the average  of the  popularity vectors associated with the users forming this cluster. Therefore, appropriately clustering  the users based on the similarity in their  popularity distributions is of  paramount importance.}
		In this paper, not all SBSs are required to  be active. We consider the density vector of active SBSs  $\Lambda_s \in \mathbb{R}^{N_c \times 1}$, where each of its coefficients $\lambda_{sk},\; k=1..N_c$ represents the density of active SBS caching the most popular files of cluster $k$. $\Lambda_s $ is defined such that $\Lambda_s^{\dagger} \mathbf{1} =\lambda_s$. 
		{ Each user  looks for the requested file in the cache of the SBSs within a radius $R$. The user starts with the closest SBS  from his own cluster. If the requested file  is available in a cache  within this distance,  a cache hit event occurs and the user will associate with the closest SBS storing the requested file. In the event of a cache miss,  the user, simply, associates with the nearest SBS from its corresponding cluster and the requested content will be retrieved from the core network  through the backhaul. If a user cannot find an SBS from its own cluster within  a radius of $R$, it  will only communicate with SBSs from other clusters within radius $R$, in the case of  a cache hit event.  An example of the considered model  is represented in Figure $1$ with three  popularity based clusters  represented, each,  by a  color}.
		\begin{figure}[!htb]
			\centering	
			\includegraphics[width=12cm,height=6cm]{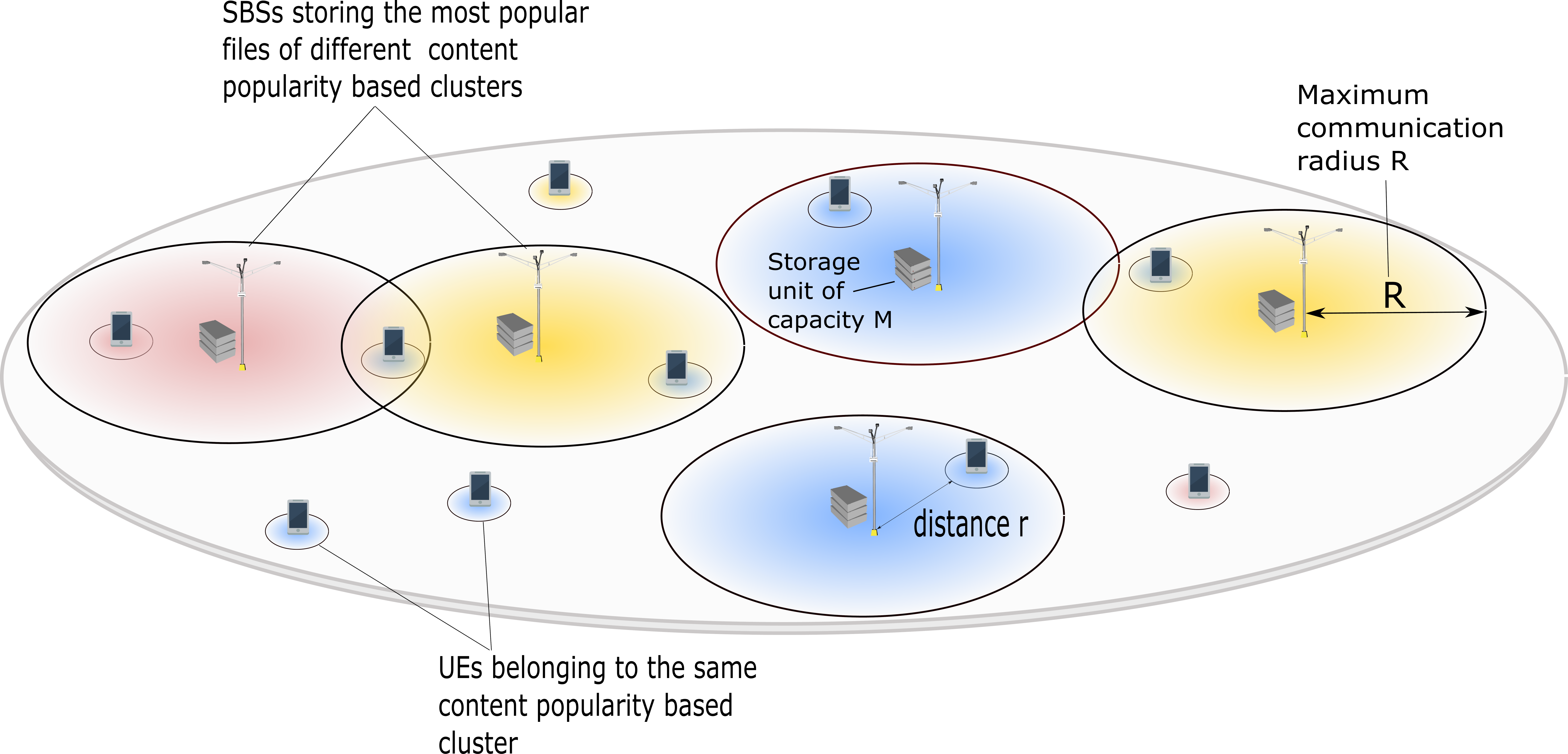}
			\label{System Model}
			\captionsetup{font=footnotesize}
			\caption{ {System Model}}
		\end{figure}
		
		\section{Information theoretic approach to user  clustering}
		Proactive caching systems require an efficient characterization of content popularity in order to  correctly predict the files that are  most likely to be requested. While most of the existing work  assume similar popularity distributions for all users, we adopt a content based clustering approach. We argue that clustering users according to their popularity distributions  enables to better assess social similarities \cite{conf_version, cont} and, consequently, devise a more efficient caching system. 	
		In fact, supposing that all users in the  network  have similar popularity distributions  means that the resulting  statistics are just an average of content  popularity over all  social groups. This  leads to  neglecting the diversity of social behavior. 
		Contrary to traditional location based clustering methods, content aware clustering enables to identify the main request patterns in the network which leads to a better understanding of user preferences. 
		In this work, content popularity based  clustering is considered. Users are grouped such that, the correlation between the  popularity profiles of users in the same cluster is maximized.  {This correlation is characterized by the euclidean  distance between their content popularity vectors.}
		While content based clustering was proposed in \cite{cont} using a spectral approach, we choose to  adopt an information theoretic method, namely, Akaike information criterion. AIC  allows to  efficiently estimate the number of  clusters and to  assess the information loss that results from assuming a single popularity distribution.
		\subsection{Cluster estimation: Akaike information criterion}
		{ 	In the  considered setting, the  users have heterogeneous popularity profiles. However, the different social relations and  interactions may result in  some correlation in user request patterns.   Consequently, content popularity based clustering is used in order to  minimize  the divergence among user content popularity distributions in each cluster.  The number of content popularity based clusters is unknown a priori and should be estimated. Allowing the system to estimate this parameter periodically or whenever a substantial  change in user interest is recorded, allows the network to cope with any modification in user request pattern.} 
		
		In order to  estimate the number of clusters, we use AIC \cite{akaike}  as a statistical model  selection criterion. It allows to assess the  quality of a statistical models for  a given set of data. The data set to be modeled in our case is the collection of user content popularity distributions. Using AIC enables to  estimate the expected Kullback-Leibler discrepancy between the data generating model and any candidate statistical model. 
		It also  addresses the trade-off between the fitness of the statistical model, based on maximum likelihood estimation, and its complexity, which is given by the number of model-characterizing parameters to be estimated.
		
		In our case, we aim at modeling the distribution that generates the user's popularity vectors. 
		We consider the true generating distribution  $A(P_1,...,P_U)=\prod_{u=1}^{U} \mathbb P_u(P_u)$, where  $\mathbb P_u(P_u)$ is the probability  that  user $u$ has a popularity vector $P_u$. We assume that  $A(P_1,...,P_U)$ results from the aggregating of  $N_c$ user clusters where, $N_c$  denotes the true number of clusters.
		Let $\xi_{N_i}, i \in [c_{min},..,c_{max}] $ be a set of  approximation  models. 	Each approximation  model $\xi_{N_i}$ is   characterized  by $N_i $ clusters and a popularity generating distribution $ \prod_{u=1}^{U} \mathbb P_u(P_u|N_i)$.  
		The popularity generating distribution depends on the  number of clusters $N_i$. In fact,  the average and  variance of content popularity vectors in each cluster depend, primarily, on $N_i$.
		The Kullback-Leibler information, which characterizes the  information lost when an approximating model is used, can be written, $\forall \; \xi_{N_i}, i \in [c_{min},..,c_{max}] $, as:
		\begin{align}\tag{2}
		& \nonumber d\big(N_i,N_c \big) = \\
		\nonumber	& \int_{[0,1]^F}^{}  \prod_{u=1}^{U} \mathbb P_u(P_u) \;\text{log}\big( \frac{\prod_{u=1}^{U} \mathbb P_u(P_u)}{ \prod_{u=1}^{U} \mathbb P_u(P_u|N_i)}\big)  \text{d} \;p_1 ... \text{d} \;p_F,\nonumber
		\end{align}
		After simplification, the  discrepancy between the two models is given by \cite{aic}:
		\begin{align}\tag{3}
		& \nonumber   d\big(N_i,N_c \big)  =  \mathbb E\left\lbrace -2 \;\text{log}\;( L_{\xi_{N_i}}(N_i| P_u, u=1...U ))\right\rbrace, \nonumber
		\end{align}
		where $  \mathbb E\left\lbrace .\right\rbrace $ denotes the expectation with respect to the available data, which is the collection of all  users popularity  profiles, knowing $N_i$.   $L_{\xi_{N_i}}(N_i| P_u, u=1...U )$ denotes the likelihood of having $N_i$ clusters, knowing the popularity profiles of the users (the expression will be given later on in this section).\\ 
		In \cite{aic}, Akaike  noted that $-2 \;\text{log} \;( L_{\xi_{N_i}}(N_i| P_u, u=1...U ))$ is a  biased  estimate of the average discrepancy. After bias adjustment,  the expected  discrepancy can be approximated by:
		\begin{align}\tag{4}
		& \nonumber \mathbb E\left\lbrace d\big(N_i,N_c \big)\right\rbrace \approx   2k_i - 2 \text{log} \;( L_{\xi_{N_i}}(N_i| P_u, u=1...U )).\nonumber
		\end{align}
		Here $k_i$ denotes the number of characterizing parameters in model $\xi_{N_i}$ and  $  \mathbb E\left\lbrace .\right\rbrace $ denotes the expectation with respect to the available data.	The   expected value of the discrepancy is asymptotically equal to the expected AIC of the considered statistical model which  is  given by:
		\begin{align}\tag{5}
		& \nonumber  \text{AIC}(\xi_{N_i})= 2k_i-2 \text{log} \;( L_{\xi_{N_i}}(N_i| P_u, u=1...U )).	\nonumber 		
		\end{align}
		AIC  allows to  assess the truthfulness of any considered statistical model, and in our case, allows to estimate the number of content based clusters together with the characterizing parameters of each one. Each cluster is  characterized by the average file popularity distribution and its variance within the cluster.
		In order to  approximate the process  generating the users probability vectors, 	we consider a set of  statistical models
		$\varXi=\left\lbrace \xi_{{N_c}_{min}} ...\xi_{{N_c}_{max}}\right\rbrace $ where, $\left\lbrace {N_c}_{min}...{N_c}_{max}\right\rbrace $ represents the range over which the search for the true  number of clusters will be carried out. Each of the considered models will be  typified by a number of defining parameters. In our case, each considered model $\xi_{N_i}$ is  characterized by 	$N_i \times (F+1)$  parameters,  $ N_i \times F $  representing the average file popularity in each cluster  and $N_i$ variance estimates.
		In this paper,  the likelihood $ L_{\xi_{N_i}}(N_i| P_u, u=1...U )$  is computed based on a Gaussian Mixture model. This is a common assumptions for data generating models \cite{xmean}.  The log likelihood  function $ \text{log} \;( L_{\xi_{N_i}}(N_i| P_u, u=1...U ))$ is computed after  clustering user with the assumption that they can be grouped into $N_i$ clusters.  $ \text{log} \;( L_{\xi_{N_i}}(N_i| P_u, u=1...U ))$ can be written as follows:
		\begin{align}\tag{6}
		\nonumber & \text{log}\;( L_{\xi_{N_i}}(N_i| P_u, u=1...U )) =\\
		\nonumber & \sum_{u=1}^{U} \big( \text{log}(\frac{1}{\sqrt{2\pi} \hat{\sigma}_{\psi(u)}^F})-\frac{\lVert P_u - \hat{P}_{\psi(u)}  \rVert^2}{2\hat{\sigma}_{\psi(u)}^2} + \text{log}(\frac{U_{\psi(u)} }{U}) \big), \nonumber 
		\end{align}
		where $\psi(u)$ represents the index of the cluster to which user $u$ is assigned.	$\hat{P}_{\psi(u)} $ denotes the average popularity vector  in cluster $\psi(u)$. $U_{\psi(u)}$ refers to the number of users in cluster $\psi(u)$. $\hat{\sigma}_{\psi(u)}^2$ denotes  the  variance of content popularity vectors in cluster $\psi(u)$ and is given by: 
		\begin{align}\tag{7}
		& \nonumber \hat{\sigma}_{\psi(u)}^2= \frac{1}{(U_{\psi(u)})} \sum_{j \in U_{\psi(u)}}^{} \lVert P_j - \hat{P}_{\psi(u)}  \rVert^2 .\nonumber
		\end{align}
		Then the log-likelihood function can be written as:
		\vspace*{-2mm}
		\begin{align}\tag{8}
		\nonumber &	\text{log} \;( L_{\xi_{N_i}}(N_i| P_u, u=1...U ))  =\\
		\nonumber &\sum_{k=1}^{N_i} 
		-\frac{U_{k}}{2}( \text{log}(2 \pi) -1 + 2 \text{log}( \frac{U_{k}}{U})- F \text{log}(\hat{\sigma}_k^2) ) .\nonumber
		\end{align}
		The resulting AIC for model $\xi_{N_i}$ is given by:
		\begin{align}\tag{9}
		\nonumber &	\text{AIC} (\xi_{N_i})  =\\
		\nonumber & 2 N_i  (F+1) + \sum_{k=1}^{N_i} 
		{U_{k}}( \text{log}(2 \pi) -1 + 2 \text{log}( \frac{U_{k}}{U})- F  \text{log}(\hat{\sigma}_k^2) ) .\nonumber
		\end{align}
		The model that best describe the user popularity vectors is the one that minimizes the AIC and, consequently, the discrepancy.
		In order to find the best model, the user are clustered according to their content popularity vectors using the $K$-mean algorithm \cite{mean}, while assuming  different numbers of clusters from  a search range $\left\lbrace {N_c}_{min}...{N_c}_{max}\right\rbrace $. The selected model  $\xi_{\text{AIC}}$ verifies:
		\begin{align}\tag{10}
		& \nonumber \xi_{\text{AIC}}=\underset{\xi \in \varXi}{\text{argmin}} \; \text{AIC}(\xi).\nonumber
		\end{align}
		The selected model which minimizes the AIC, strikes the best trade-off between fitness and complexity. This results in a truthful modeling of content popularity based clusters. The resulting model guarantees minimum discrepancy among the request patterns of the users within each cluster. We now provide the detailed description  of the content  based user clustering algorithm.
		\subsection{User clustering algorithm}		
		The proposed { content popularity} based clustering algorithm starts by defining a search interval $\left[{N_c}_{min}...{N_c}_{max}  \right] $.  The algorithm  begins by assuming the existence of  ${N_c}_{min}$ clusters. It  clusters the users accordingly using the $K$-mean algorithm \cite{mean}.  
		$K$-mean  allows to  assign each user to the  cluster with the nearest centroid  which results in  minimizing the disparity between users behaviors in the same cluster. The popularity profile of the cluster is then defined as the average of the  popularity vectors of all users in the cluster as:
		\begin{align}\tag{11}
		& \nonumber \hat{P}_{k} = \frac{\sum_{u ,\psi(u)=k }^{} P_u  }{U_{k}}.\nonumber
		\end{align} 
		{Each cluster $k$ is then associated with a vector $\hat{P}_{k}=\left[ \hat{p}_{1k}...\hat{p}_{Fk}\right] $, where  $\hat{p}_{fk}$ denotes the average popularity of file $f$ in cluster $k$.}
		Once  users are assigned to their respective clusters,  AIC is computed.
		The number of clusters is incremented by adding a new centroid. The AIC is then recomputed until reaching a minimum.
		AIC is decreasing as a function of the number of clusters until reaching a minimum in the most accurate estimate. The AIC will then start increasing because of model complexity. 
		Since the goal of the clustering is to reduce  the divergence among users from the same cluster, a new centroids is added, at each step, in the cluster with the greatest popularity  variance.
		The new center is selected as the user having the largest distance from the mean popularity vector of its cluster.
		This  allows  to  reduce the  discrepancy in user traffic  pattern.
		The detailed clustering algorithm can be written as the following:
		\begin{center}
			\begin{tabular}{ l  }
				\textbf{Content-popularity based user clustering algorithm}\\
				\hline
				\hline
				\emph{Initialize}: Define search interval $\left[{N_c}_{min}...{N_c}_{max} \right]$, Set\\  $K={N_c}_{min}$
				Choose randomly the first ${N_c}_{min}$ centroids\\ from the available users\\
				$1. $ Run $K$-mean algorithm and compute $\mathrm{AIC}(\xi_{K})$\\
				$2. $ Choose the user having the largest distance from its\\ \quad centroid in  cluster $k^*$  with  the greatest variance \\ \quad { ($k^*=\underset{k=1...K}{\text{argmax}} \; \hat{\sigma}_{k}^2$)}\\ 
				$3. $ Add a centroid with the popularity profile of \\ \quad the chosen user and set $K=K+1$\\
				$4. $ {Run step $1$ to step $3$ }until AIC starts to increase.\\
				$5. $ Choose the model which minimizes the AIC 	and \\ \quad cluster the users accordingly\\ 	
				\hline
				\hline
			\end{tabular}
		\end{center}
		\vspace*{5mm}
		
		{Once content popularity  based clustering is performed, the cached files of each cluster are selected based on its average   popularity vector, which is given in $(11)$. 
			For each $\{k=1...N_c\}$,  the files in the catalog are ordered in a decreasing order of popularity according to $\hat{P}_{k}$. 
			The set of cached files, in each cluster, $\{\Delta_k, k=1...N_c\}$ is then  selected as the most popular files,  according to $\{\hat{P}_{k}, k=1...N_c\}$, whose aggregate size is at maximum $M$.
			Apart from the  maximum size  constraint, we impose no restrictions on the set of cached files. 
			Consequently, there may be some overlapping between the cached files of different clusters. Meaning that the same file can be selected in the cached sets of different clusters ($\Delta_k \cap \Delta_{j} \neq  \emptyset, \text{for some}  \; k \neq j$).
			Files that are selected by different  clusters are very popular across  user. Consequently, it makes sense to increase the number of cached copies in the network. Given the considered model in section $II$, allowing overlapping between the cached files of different clusters provides better performance.
			Since user preference can change over time,   the algorithm in Table $I$ can be executed periodically or whenever substantial  popularity  profile  modification is recorded. This  allows the caching system  to  adapt the selected files  accordingly. In order to investigate the performance of the proposed scheme,  we assess its impact on the achievable EE of the system. }
		\section{EE with content popularity clustering}
		Predicting which content  is most likely to be requested and caching it
		in the edge can reduce the latency and backhaul load  as well as increasing the overall throughput.  Proactive caching also proved to be an effective technology that can improve another very important metric in future generation networks, namely, EE \cite{EED}. In what follows, we investigate the EE of cache enabled small cell networks with content popularity based user clustering. 
		We consider the downlink of the cache enabled network. Without loss of generality,  we  concentrate on a reference user located at the origin of the plane. The EE of the network is given by  the ratio between the average achievable spectral efficiency and the average consumed power \cite{EED}.
		\begin{align}\tag{12}
		&\nonumber \Sigma=\frac{SE}{\rho^c_{total}},\nonumber
		\end{align} 
		where $\Sigma$ denotes the average energy efficiency, $\rho^c_{total}$ denotes the average consumed power in the cache enabled small cell network and  $SE$ denotes its average achievable spectral efficiency.
		In order to  derive the expressions of  $SE$ and $\rho^c_{total}$ and, consequently, the achievable EE, we need to start by finding  the expression of the cache hit probability. 
		\subsection{Cache hit Probability }
		According to the considered system model, the cache hit probability refers to the probability of finding a requested file in the cache of a SBS within radius $R$ from a given user \cite{modeling}. 
		Our context is different  from the one in \cite{modeling}, since the users are clustered and  the SBSs cache different files depending on their associated cluster. Considering the proposed clustering model, the cache hit probability can be expressed as follows (the derivations are skipped for brevity):
		\begin{align}\tag{13}
		\nonumber \mathbb P \left\lbrace  hit\right\rbrace  = \frac{1}{U} \sum_{k=1}^{N_c} \sum_{u=1}^{U} \big( \sum_{i \in \Delta_k}^{} p_{iu} \big) \big( 1- e^{-  \lambda_{sk}\pi R^2}\big),  \nonumber
		\end{align}
		where  $\Delta_k$ represents the set of the  most popular files of cluster $k$ that fills the SBS caching capacity.
		This equation denotes the probability of finding of at least one SBS with the requested file stored in its cache   within a radius $R$ from a given user.  The density of SBS caching the most popular content of a cluster $\{k, k=1...N_c\}$ is given by $\lambda_{sk}$   and, their average number $N_{sk}$ is given by $N_{sk}=\lambda_{sk}\pi R^2_n$.  The densities $\{\lambda_{sk}, k=1...N_c\}$ are such that $\sum_{k=1}^{N_c}\lambda_{sk}=\lambda_{s}$. One major upside of content popularity based user clustering is content diversity. While each SBS caches the most popular files of  only one cluster,  users can request any of the  cached files  in  SBSs within radius $R$, which can be fetched without additional load on the backhaul. 
		In fact, a given user can communicate with  the closest SBS caching the  files of a cluster different from his own whenever the requested content is already cached.
		Consequently, compared with the classical approach of caching the same popular content everywhere,  the users covered by several SBSs from different clusters will see an increase in their  cache hit probability.	
		\subsection{Average total consumed power}
		In order to gain a useful insight into the achievable EE  and capture the fundamental tradeoffs, we extend the power model in \cite{EED}.
		The average consumed total power in the considered network with caching capabilities can be modeled as follow:
		\begin{align}\tag{14}
		\nonumber	& \rho^c_{total}= \mathbb E \left\lbrace \rho_I\right\rbrace + 
		\mathbb E \left\lbrace \rho_T\right\rbrace +\mathbb E \left\lbrace \rho_f\right\rbrace,\nonumber	
		\end{align}
		where $\rho_I$, $\rho_T$ and $\rho_f$ denote,  respectively, the power consumed by the infrastructure of active base stations, the total transmit power and the used power to fetch files from the hard disc or the core network. The expectation  $\mathbb E \left\lbrace .\right\rbrace$ is taken over  the users and SBSs PPPs. \\	
		The average  power consumed by the infrastructure is given by:
		\begin{align}\tag{15}
		\nonumber	& \mathbb E \left\lbrace \rho_I\right\rbrace = \rho \lambda_s \pi R^2_n,	\nonumber
		\end{align}
		{	$\rho$ and  $\lambda_s \pi R^2_n$ denote, respectively, the fix operational charge consumed by an active SBS and the average number of active SBSs. The average power used to retrieve a file either over the backhaul, when a cache miss event  occurs, or from a SBS cache is given by:
			\begin{align}\tag{16}
			\nonumber & \mathbb E \left\lbrace \rho_f\right\rbrace =\lambda_s \pi R^2_n\big(  \rho_{hd}  \mathbb P \left\lbrace hit\right\rbrace +  \rho_{bh} \big( 1-\mathbb P \left\lbrace hit\right\rbrace\big) \big),\nonumber  
			\end{align}
			where $\rho_{hd}$ denotes the power needed to retrieve data from the local hard disk of a small  base station when the  requested content is already cached and a cache hit event occurs. $\rho_{bh}$ denotes the power needed to retrieve data from the core network through the backhaul when a cache miss event  occurs.
			Owing to channel inversion power control, the power used for transmission depends on the distance between  the  communicating SBS and users. 
			Here we consider $\Upsilon_k $ as the set of users associated with cluster $k, \forall k=1..N_c$.
			{	Each user looks for the requested file in the cache of the SBSs within a radius $R$, starting with the closest SBS  from his own cluster. 
				If the requested file  is available in a cache  within this distance,  a cache hit event occurs and the user will associate with the closest SBS storing the requested file. In the event of a cache miss,  the user  associates with the nearest SBS from its corresponding cluster and the requested content is retrieved from the core network  through the backhaul. If a user cannot find an SBS from its own cluster within  a radius  $R$, it   only communicate with SBSs from other clusters within radius $R$, in the case of  a cache hit event.} The average total transmission power is given by:
			\begin{align}\tag{17}
			\nonumber & \mathbb E \left\lbrace \rho_T\right\rbrace = \frac{\lambda_s \pi R^2_n}{U} \sum_{k=1}^{N_c}  \sum_{u \in \Upsilon_k}^{}  \big(    \mathbb E \left\lbrace \rho_k\right\rbrace\\
			\nonumber & + \sum_{j \neq k}^{} \sum_{i \in \Delta_s}^{} p_{iu}  ( 1- e^{-  \lambda_{sj}\pi R^2})
			(\mathbb E \left\lbrace \rho_j\right\rbrace  -\mathbb E \left\lbrace \rho_k\right\rbrace)\big),\nonumber
			\end{align}
			where  $E \left\lbrace \rho_k\right\rbrace$ denotes the average transmit power that a typical user utilizes  when communicating with  SBSs associated with cluster $k, \forall k =1..N_c$. }
		Finally, the  expression of the  average consumed total power is derived by including the expressions  of $\mathbb E \left\lbrace \rho_k\right\rbrace$, $\forall k=1..N_c$.
		\begin{lem}  		
			The average consumed total power in the considered network with caching capabilities and content-popularity based user clustering can be modeled as follows:
			\begin{align}\tag{18}
			& \rho^c_{total}\nonumber  =\lambda_s \pi R^2_n(  \rho_{hd}  \mathbb P \left\lbrace hit\right\rbrace +  \rho_{bh} \big( 1-\mathbb P \left\lbrace hit\right\rbrace\big)  +\rho)\\
			\nonumber & +\frac{\lambda_s \pi R^2_n}{U}  \sum_{k=1}^{N_c}  \sum_{u \in \Upsilon_k}^{}( \frac{\rho_0 \gamma(\frac{\alpha}{2}+1,\pi \lambda_{sk} R^2)}{(\lambda_{sk}\pi)^{\frac{\alpha}{2}} }\\
			\nonumber &  + \sum_{j \neq k}^{} \sum_{i \in \Delta_s}^{} p_{iu}  ( 1- e^{-  \lambda_{sj}\pi R^2})(\frac{\rho_0 \gamma(\frac{\alpha}{2}+1,\pi \lambda_{sj} R^2)}{(\lambda_{sj}\pi)^{\frac{\alpha}{2}} }\\
			\nonumber &  -\frac{\rho_0 \gamma(\frac{\alpha}{2}+1,\pi \lambda_{sk} R^2)}{(\lambda_{sk}\pi)^{\frac{\alpha}{2}} })).
			\end{align}	
		\end{lem}
		\begin{proof}
			See Appendix A.
		\end{proof}	
		Following the same reasoning, the average consumed total power in the network  with no proactive caching capabilities at the  SBSs, is given by:
		\begin{align}\tag{19}
		\nonumber	\rho^{nc}_{total}= \lambda_s \pi R^2_n \rho_{bh} +\rho \lambda_s \pi R^2_n + \lambda_s \pi R^2_n\frac{\rho_0 \gamma(\frac{\alpha}{2}+1,\pi \lambda_{s} R^2)}{(\lambda_{s} \pi)^{\frac{\alpha}{2}} }.\nonumber
		\end{align}
		$\rho^{nc}_{total}$ is taken into consideration  in order to guarantee an improvement in the  average EE of the network,  when proactive caching is implemented.
		\subsection{Average Spectral Efficiency}
		In order to  derive the achievable EE, the expression of the average spectral efficiency should be derived.  
		The downlink SINR for a  user $u$ taken at the origin is given by:
		\begin{align}\tag{20}
		\nonumber &  \text{SINR} = \frac{\rho_0 \left\|  h_u\right\| ^2}{\sigma^2 + \sum_{k=1}^{N_c} I_{k}},\nonumber	
		\end{align}
		where $I_k,\; \forall k=1..N_c$ represents the interference coming from SBS from cluster $k$ given by $ I_{k} =   \sum_{i \in \phi_{sk}}^{} \rho_{ik} \lVert h_{ui}\rVert^2 r^{-\alpha}_{ui} $. 
		Here  $\phi_{sk}$ denotes the  set of SBSs associated with cluster $k,\; k=1..N_c$. $\rho_{ik}$ refers to the power used in the downlink by SBS $i$ from cluster $k$. $\sigma^2$ represents the  noise power.
		In order  to compute the average spectral efficiency of the  network, first we need to derive the achievable coverage probability which is given  in the  following Lemma:
		\begin{lem} 
			The downlink coverage probability is given by:
			\begin{align}\tag{21}
			\nonumber &  \mathbb P \left\lbrace \text{SINR}  \geq \theta \right\rbrace = exp (- \frac{\theta}{\rho_0} \sigma^2) \times\\
			\nonumber &	\prod_{k=1}^{N_c} exp \big(-{\pi \lambda_{sk} }\Gamma(1+\frac{2}{\alpha}) \Gamma(1-\frac{2}{\alpha}) (\frac{\theta}{\rho_0})^{\frac{2}{\alpha}} \E{\rho^{\frac{2}{\alpha}}_k} \big).\nonumber	
			\end{align}
		\end{lem}
		\begin{proof}
			See Appendix B.
		\end{proof}		
		We can see from Lemma $2$, that increasing the SBS density enables to reduce the used transmit power. Nevertheless, we need to take into consideration the constant power consumed by the infrastructure of active SBSs  which represents an important part of  power consumption of the network. 
		The  average achievable spectral efficiency can be written  as:
		\begin{align}\tag{22}
		\nonumber &  SE = \lambda_s \pi R^2_n \; \text{log} (1+\theta) \; \mathbb P \left\lbrace \text{SINR}  \geq \theta \right\rbrace.\nonumber
		\end{align}
		Given the average achievable spectral efficiency  and average consumed power, we can derive a closed form expression of the energy efficiency $\Sigma$: 
		\begin{align}\tag{23}
		\nonumber & \Sigma= \frac{\lambda_s \pi R^2_n \; \text{log} (1+\theta) \; \mathbb P \left\lbrace \text{SINR} \geq \theta \right\rbrace}{\rho^{c}_{total} }.	\nonumber
		\end{align}
		By substituting $(18)$ and $(21)$ into $(23)$, we obtain the average EE. We can notice from $(21)$ that the density of SBSs is a major defining parameter of $\Sigma$. 
		\subsection{Analysis of  Energy Efficiency}
		We can see, in $(22)$, that  increasing the SBS density   results in a reduction in the interference. This is mainly due to the resulting decrease in transmit power since users are closer to their serving SBSs. Nevertheless, increasing SBS density   results in more power consumption due to the strain of active infrastructure. We aim  at finding the optimal active SBS density vector that maximizes the achievable EE, even when user positions are not taken into consideration. We can imagine a setting in which SBS are activated and shutdown based on user density and popularity profiles. 
		We consider a constraint in which we aim at maintaining a power budget that is lower than that  used when no proactive caching is enabled. The  problem  can be  formulated as follows:
		\begin{align}\tag{24}
		\underset{\Lambda_s}{\text{maximize}} &  \;\; \Sigma \\\nonumber \tag{24a}
		\text{subject to}	&  \;\;\rho^{c}_{total}-\rho^{nc}_{total}\leq 0,\\\nonumber \tag{24b}
		& \;\;\Lambda_s^{\dagger} \mathbf{1} \leq \lambda_{s_{max}}.\nonumber
		\end{align}
		This optimization problem  allows to  derive the optimal  density vector  needed to maximize the average EE for a given user density, popularity profiles and cache size.  Although the closed form expression of $\Sigma$ is complex to analyze, it can be proven  that $\Sigma$ is quasi concave using an intelligent simplification  by considering a composition of $\Sigma$  with an affine mapping \cite{boyd}.  
		\begin{Theorem}
			The considered optimization problem  is quasi concave  and the optimal SBS density $\Lambda^*_s$ can be derived, with zero duality gap. 
			If $\exists \Lambda^*_s$ such that $\nabla \Sigma(\Lambda^*_s)=0$ then this vector is unique and it is the optimal solution. If this condition is not satisfied for any $\Lambda_s$ such that $\Lambda_s^{\dagger} \mathbf{1} \leq \lambda_{s_{max}}$, then the optimal solution can be found using the   Karush-Kuhn-Tucker (KKT) conditions:
			\begin{equation*}\tag{25}
			\begin{aligned}
			& \nabla L(\Lambda^*_s,\varsigma,\kappa)= \nabla \Sigma(\Lambda^*_s)+ \varsigma \nabla C(\Lambda^*_s) +\kappa \nabla H(\Lambda^*_s)=0,\\
			& \varsigma C(\Lambda^*_s)=0, \kappa H(\Lambda^*_s)=0,\\
			& H(\Lambda^*_s) \leq 0, C(\Lambda^*_s) \leq 0,\\
			& \varsigma>0,\kappa>0, \\
			\end{aligned}
			\end{equation*}
			where $L$ refers to the Lagrangian associated with problem $(24)$, $C(\Lambda_s)=\rho^{c}_{total}-\rho^{nc}_{total}$, $H(\Lambda_s)=\Lambda_s^{\dagger} \mathbf{1}$.  $\varsigma\; \text{and}\; \kappa$ refers  to  the  Lagrangian multipliers associated, respectively, with  $(24a)$ and $(24b)$.
		\end{Theorem} 
		\begin{proof}
			See Appendix C.
		\end{proof}	
		Finding the  optimal  density vector $\Lambda^*_s$ based on the  KKT conditions in  $(25)$, can be  done using, for example,   the  sub-gradient descent method \cite{boyd}. 
		%	The derivation of optimal solutions using KKT conditions has been widely utilized in the literature and is skipped here for brevity.	
		% Although $\Lambda^*_s$ provides the optimal  SBS density that is associated with each cluster, further improvement can be made by considering the allocation of each individual SBS.
		%In fact, caching files closer to  the users that are most likely to request theme will result in lower transmit power, thanks to channel  inversion power control. It enables to  reduce the  interference in the network resulting  in higher spectral  efficiency. In order to enable the optimization of SBSs allocation to their respective clusters, a  snapshot of the  users and SBSs PPPs will henceforth be considered. 
		\section{Exploiting spatial correlation in users demand}
		{		 While, in the previous sections, EE was optimized with respect to the  density vector of active SBSs, further EE gain can be  achieved by including spatial  information whenever it is available.   
			In this paper, the choice to decouple the two problems of cached file selection and content placement can be justified by the fact that popularity distributions change slower than user locations.   Consequently, the selected cached content which depends on the average  popularity distribution per cluster, is kept  constant for long periods and the network can adapt its location based on user movement. Real life examples can also support this approach.  While correlation in content popularity  between users from the same social  group stays for long periods of time, their location can  change due to  mobility. This motivates the need to adapt the cached content placement more often than the selected files in order to simplify the management of the network. Acquiring  information on user location requires a non negligible processing and signaling overhead. Consequently, this information should be leveraged whenever it is available. 
			
			In the case of low mobility, where users do not change positions too often, it makes sense to adapt the files placement based on location information.  	Adapting the cached files placement   can be done  periodically  or whenever the  backhaul load allows it.
			
			Practically, we may observe  a  spatial correlation in user file demand. This can be explained by the fact that people from the same social group (living or working in the same place) are most likely to have similar preferences. We aim at finding an  effective allocation of the  SBSs to the different  clusters in order to  minimize transmit power  and, consequently, to improve the achievable EE.  In fact,  decreasing the distance between a given user and  the SBS storing its requested file results in lower transmit power. Consequently, by effectively allocating the SBSs to the different clusters, we are able to  lower the level  of interference in the network, which results in increasing the EE \cite{EED}. The problem  of cache placement  can be tackled  by adopting a hybrid approach  where, a clustering based on both the location and  content popularity  is performed. This approach  is more complex and do not necessarily produce better results since  content popularity stays constant for long periods of time.  In addition,  the clustering that is done on the users enables also to group the files accordingly. Consequently, the complexity of the resulting optimal file placement  problem is lower since the search space is reduced from the whole file catalog to groups of file of approximately equal total size. This  simplifies the management of the caching system compared to  existing work on location based optimization where, the complexity of the formulated problems is proportional  to  the  number of files.
			
			We consider a setting in which the location of all  users and  SBSs are known. This is implemented by considering a  snapshot of the  users and SBSs PPPs.
			We develop an integer optimization problem where we aim at minimizing the used power over the possible SBS-cluster affectation. Thanks to channel inversion power control, minimizing the used power is equivalent to reducing the distance between the users and the  SBS  caching the files they are most likely to request. 
			We define $\omega_{u,s}$, the weight of the link between  user $u$  and the  SBS $s$ as follows:
			\begin{equation}\tag{26}
			\begin{aligned}
			& \omega_{u,s}   =
			\left\{
			\begin{array}{ll}
			r_{us}^{-\alpha}  & \mbox{if $r_{us} < R$,}  \\
			\omega_{\infty} & \mbox{otherwise, } 
			\end{array}
			\right.\\
			\end{aligned}
			\end{equation}	
			where $r_{us}$ representing the distance between user $u$  and the  SBS $s$. $\omega_{\infty}$ is  an arbitrarily  large value. $\omega_{\infty}$ assures that no user can communicate with a SBS at a distance larger than $R$. 
			We sort the links pathloss coefficients in decreasing order and denote by  $(s)_u$ the SBS with  the $s$-th greatest pathloss coefficient to user $u$. 
			Based on the considered system model, less power is used  when a user is served from a SBS within its  neighborhood. Consequently, maximizing  $\omega_{u,s}$ is equivalent to minimizing the transmit power  and the distance between the user and its serving SBS. The average number of SBS associated with each cluster $k$ is given by $N_{sk}= \lambda_{sk} \pi R^2_n$. Since $\lambda_{sk},\; k=1..N_c$ are computed in $(24)$ so that EE is maximized, it does not take into consideration their spatial repartition  in the network.
			The transmit power increases when users from the  same cluster are  not located within a reduced area. In order to  deal  with  this problem, we  relax the  constraint on the  number  of SBS  per cluster and  we replace $N_{sk}$ by $N'_{sk}$ where $N'_{sk}>N_{sk}$. We  consider the adjacency matrix $Y$,  where $y_{s,k}, \forall s=1..N_s, k=1..N_c$ is given by:
			\begin{equation}\tag{27}
			\begin{aligned}
			& y_{s,k}   =
			\left\{
			\begin{array}{ll}
			1  & \mbox{if  SBS $s $ is associated with  cluster $k$.}  \\
			0 & \mbox{otherwise. } 
			\end{array}
			\right.\\
			\end{aligned}
			\end{equation}
			The problem of optimal SBS allocation to their respective clusters can be formulated  as:
			\begin{align}\tag{28}
			&	\underset{Y}{\text{max}}  \sum_{k=1}^{N_c} \sum_{s=1}^{N_s} \sum_{u=1}^{U}   \big( \sum_{f \in \Delta_k}^{}  p_{fu}\big) \omega_{u,(s)_u}( y_{(s)_u,k}  \prod\limits_{i=1}^{s-1} (1- y_{(i)_u,k}   )) \\\nonumber \tag{28a}
			&	\text{subject to}   \sum_{k=1}^{N_c} y_{s,k} \leq 1 , \forall s= 1..N_s,\\\nonumber \tag{28b}
			& \sum_{s=1}^{N_s} y_{s,k} \leq N'_{sk} , \forall k= 1..N_c.\nonumber
			\end{align}
			Here, $(28a)$ captures the fact that each SBS stores the most popular files of one unique cluster. $(28b)$ indicates that the number of SBSs allocated to each cluster should  respect the density vector $\Lambda^*_s$ which maximizes EE. 
			The objective function in $(28)$ guarantees that  each SBS caches the  files that are most likely to  be requested by nearby users.  In fact, $( y_{(s)_u,k}  \prod\limits_{i=1}^{s-1} (1- y_{(i)_u,k}   ))$ is an indicator function that  refers to the case where the  most popular files of  cluster $k$ are cached in  SBS $(s)_u$ and not in  SBSs $(i)_u, i=1,...,s-1.$ Consequently, the objective function value is equal to  the expected pathloss between users and their serving SBSs. 
			We show that the considered optimization problem is NP-hard. We then prove that it can be formulated as the maximization of a  submodular function over matroid constraints and we provide an  algorithm that enables to  derive a $(1-\frac{1}{e})$ approximation of the optimal  solution of problem $(28)$. 
			This formulation looks somehow similar to  the  considered problem  in \cite{femto} where, the authors aim at optimizing the  allocation of each  individual  file to a set of femto access points in order to  minimize the  expected downloading  time.
			Nevertheless, problem $(28)$ consider a different  objective function where, the aim is to  minimize the  transmit power. While \cite{femto} aims at optimizing the assignment of each individual file to the different femto access points, the objective in $(28)$ is, actually, to assign predefined  batches of files from each cluster to the SBSs. Consequently, the problem formulation in the present paper enables to considerably reduce the complexity of deriving a solution. In fact, the running time  depends on the number of popularity based clusters rather than the number of files. This is an  important  impact of the present formulation in $(28)$ since the number of files is typically very large. The considered setting enables to solve problem  $(28)$ using sophisticated algorithms that can be computationally prohibitive  otherwise.}
		\subsection{Computational Intractability}
		We start  by showing the computational intractability of problem $(28)$. 
		\begin{Theorem} 
			The considered optimization problem in $(28)$ is NP-hard.
		\end{Theorem}
		\begin{proof}
			In order to show that $(28)$ is NP-hard, we consider a special case of our setting where $N'_{sk}= N \; \forall k= 1..N_c$ and $N_s=N_c$. This special case means that the number of SBS associated with each cluster is the same, which is the case when  $\sum_{u=1}^{U}\sum_{f \in \Delta_k}^{}  p_{fu} = C, \; \forall k= 1..N_c$. 
			In this case, the resulting optimization problem can be written as follow:
			\begin{equation*}\tag{29}
			\begin{aligned}
			\underset{Y}{\text{max}} &  \sum_{k=1}^{N_c} \sum_{u=1}^{U}\sum_{s=1}^{N_s} C\omega_{u,(s)_u}( y_{(s)_u,k}  \prod\limits_{i=1}^{s-1} (1- y_{(i)_u,k}   )) \\
			\text{subject to} & \sum_{k=1}^{N_c} y_{s,k} \leq 1 , \forall s= 1..N_s,\\
			&  \sum_{s=1}^{N_s} y_{s,k} \leq N , \forall k= 1..N_c.\\
			\end{aligned}
			\end{equation*}
			In order to show NP-hardness, we  use a reduction from the following NP-hard problem:\\ 
			\emph{Weighted $K$-Set Packing Problem}:
			$K$-Set packing  is an NP-hard combinatorial problem. It is one of the 21 problems of Karp \cite{packing}. The $K$-Set packing problem  aims to  find a maximum number of  pairwise disjoint sets, with at  most $K$ elements, in a family $S$ of subsets of a universal set $V$.
			The  weighted version of the $K$-Set packing problem is obtained by assigning  a real weight to each subset and maximizing the total weight.\\
			We consider a collection of SBS sets $\{v_i,i=1,...,n\}$, associated each with a weight $\omega_{v_i}= \sum_{u=1}^{U} \underset{s \in v_i}{\text{max}} \;\omega_{u,s} $. Problem $(28)$ can then be  formulated as a Weighted K-Set Packing Problem:
			\begin{equation*}\tag{30}
			\begin{aligned}
			\underset{X}{\text{maximize}} & \sum_{i}^{}  C \omega_{v_i} x_i\\
			\text{subject to} & \; v_i\cap v_j= \emptyset, \forall i,j,\\
			& \left| v_i\right|  \leq N , \forall k= 1..N_c,\\
			& x_i \in \{ 0 , 1 \}.\\
			\end{aligned}
			\end{equation*}
			Solving $(30)$  results in at most $N_c $  sets of SBSs.  Since the resulting sets are disjoint, each of them will be  associated with a given cluster. The number of resulting sets could not exceed $N_c$ since $N_s=N_c$.
			We  can see that solving the weighted $K$-Set Packing Problem, for $K=N$ and where the weight of each subset is  given by $C \omega_{v_i}$, is equivalent to  solving the special case of the SBS allocation problem in $(29)$. Knowing  that the Weighted $K$-Set Packing Problem is NP-hard, we can then conclude that  $(28)$ is also NP-hard.
		\end{proof}
		\subsection{Optimizing small base station allocation}
		{	In order to solve the considered optimization problem in $(28)$, we start by showing that it is equivalent to the maximization of a sub-modular set function over matroid constraints. The  definitions of matroids and sub-modular set functions can be found in \cite{comb}.
			This  structure allows the use the randomized algorithm proposed in \cite{subopt} which achieves, at least,  $(1-\frac{1}{e})$ of the optimal value. Taking into consideration the  problem constraints we have the  following:	
			%\begin{Definition}
			%	Let S be a finite ground set. A submodular set function $f: 2^S\longrightarrow R$ is a function that verifies for all 
			%	$A, B \subseteq S$: $f(A) + f(B) \geq f(A \cup B) + f(A \cap B).$ 
			%	Intuitively, submodular functions capture the concept
			%	of diminishing returns.	
			%	\end{Definition}
			%	\begin{Definition}
			%	A matroid $ M$ is a tuple $M = (S, I)$, where $S$ is a finite ground set and $I \subseteq 2^S$
			%	(the power set of S) is a collection of independent sets, such that:
			%	\begin{itemize}
			%		\item  $I$ is nonempty, in particular.
			%		\item $I$ is downward closed; i.e., if $Y \in I $and $X \subseteq Y $,
			%		then $X \in I$
			%		\item If $X, Y \in I $ and $\lvert X \rvert < \lvert  Y \rvert$ then $\exists\; y \in Y \setminus X $ such
			%		that $X \;\cup y \in I$	
			%	\end{itemize}	
			%	\end{Definition}
			\begin{lem}
				The Considered Optimization problem in $(28)$ is equivalent to a maximization of a sub-modular set  function over matroid constraints.
			\end{lem}
			\begin{proof}
				See Appendix D.
			\end{proof}	
			In order to solve the considered problem, we use  the randomized algorithm proposed in \cite{subopt}. This algorithm provides a $(1-\frac{1}{e})$-approximation of the optimal  solution for sub-modular set function maximization with matroid constraints. This algorithm   consists  of  two steps. In the first one,  a fractional solution of the  relaxed problem is obtained using  a continuous greedy process.
			In the second part of the algorithm, the derived fractional solution is rounded using a variant of the pipage rounding technique \cite{pip}.
			In a typical setting, randomly rounding a fractional solution of an optimization problem does not preserve the feasibility of the solution, in particular when equality constraints are considered.  Nevertheless, the pipage rounding technique in \cite{pip} enables to round a fractional solution  so that the problem constraints are not violated.
			In our case the running  time of the algorithm  is  $O((N_s N_c)^8)$ \cite{subopt}, where $N_s= \lambda_s \pi R^2_n$. This is quite convenient since the running time of the algorithm does not depend on the number of files in the catalog which can be very large.  This an interesting result of content based clustering since it reduces the search space from the whole catalog to bins  of files of approximately equal size. }
		%	in  order to cope with user mobility. 
		\section{Numerical Results AND Discussion}
		{In this section, we  investigate the impact of the different system parameters on the Cache hit probability and  EE. We then investigate the impact of the SBS allocation algorithm on the performances of the network. We consider a circular region with  an area of  $A = 10 Km^2$.
			We simulate two PPP processes, one for the users and another for the SBSs  over this  area. The respective  densities of these process are $\lambda$ and $\lambda_{s_{max}}$ with $\lambda>>\lambda_{s_{max}}$ .
			The considered SBS density values  are defined based on the typical communication range of a SBS.
			We consider a catalog  constituted of $F=2000$ files with different randomly generated sizes $L_i, \; i=1...F$ in the range $\left[10 \;\text{MB}...100\; \text{MB}\right]$ \cite{EED}. We also consider the normalized cache size $\eta=\frac{M}{\sum_{i=1}^{F}L_i}$.
			We characterize each cluster by a given popularity based  file ordering.  For each user $u$, $P_u$   is generated  according to a Zipf  distribution  with  parameter $1$ \cite{MATHA}, after randomly selecting a  cluster file ordering.  
			This  results in a random allocation of the users to the different clusters.
			In order to run the clustering algorithm, we only need an interval over which the search of the number of clusters is carried on.  In our simulations we take $\left[{N_c}_{min}...{N_c}_{max} \right] = \left[5...30 \right]$. We consider a pathloss exponent  $\alpha=2.5$.
			We consider the following power values \cite{EED}:}
		\begin{center}
			\begin{tabular}{|c|l|c|l|}
				\hline
				$\rho_o (dBm)$& $21$ & $ \rho_{bh}(W)$ & $10 W$\\
				\hline
				$\rho (W)$& $10.16 $ & $\rho_{hd}(W)$ & $12.5 \times 10^-5 $\\
				\hline
			\end{tabular}
			\label{tab2}
		\end{center}

		\begin{figure}[!htb]
			\centering	
			\includegraphics[scale=0.8]{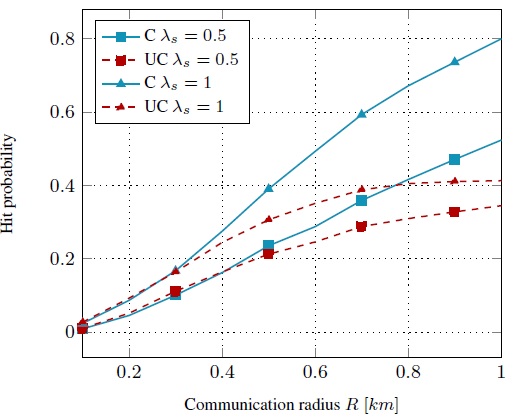}
			\captionsetup{font=footnotesize}
			\caption{Hit probability versus communication radius with different SBS densities (C= content based clustering, UC= Unclustered approach), Normalized Cache Size $\eta=0.25$}
		\end{figure}

		Figure 2 shows the evolution of the cache hit probability as a function of the communication radius $R$  for different SBS densities.  Figure 2 also shows a comparison between  the  content popularity  based clustering approach  and the classical method of supposing the same content popularity  among all users.
		As an  example, for a communication radius of $0.9\; Km$ and a SBS density of  $\lambda_s=1$, we notice a substantial  gain with  a cache hit probability of $0.736$ for the clustering scheme compared  with a probability equal to $0.41$ when caching the most popular files in all SBSs.
		We notice that the hit probability for the scheme without clustering saturates at a low value. This is due to the fact that the same set of files is cached in all the SBSs, which is clearly a suboptimal approach, especially when users are covered by multiple SBSs. The increase  in hit probability for the clustering method is  mainly due to the diversity of files cached in the SBS. Increasing the SBS  density results in reducing the  average distance from mobile users, which, consequently, results in improving the cache  hit probability.

		\begin{figure}[!htb]
			\centering	
			\includegraphics[scale=0.8]{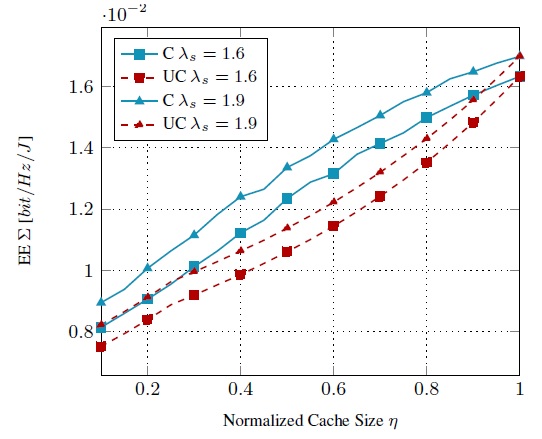}
			\captionsetup{font=footnotesize}
			\caption{ EE   vs normalized cache size with different SBS densities (C= content  based clustering, UC= Unclustered approach)}
		\end{figure}

		{	Figure 3 shows the evolution of the achievable EE as a function of the normalized cache size for different SBS densities.  Figure 3 also shows a comparison between  the  content popularity  based clustering approach  and the classical method of supposing the same content popularity  among all users.
			We can see that the proposed clustering method outperforms the classical  approach of caching the same most popular files in all SBSs.
			For a   normalized cache size of $0.4$ and an SBS density of $1.6\; \text{SBS}/Km^2$, we notice an increase of  $12.5\%$ in the achievable EE. This  gain is mainly due to the  fact  that  the proposed method scores a higher hit probability  than the unclustered approach.  Consequently, the average  energy needed to  fetch  the requested content  is lower when  user clustering  is  used.  Even though  restricting users to communicate with  the closest SBS  from  their cluster, in  the case of a cache miss event, can lead to  an increase in the average transmit power, the observed gain in the energy used to fetch the desired content compensates for that.  
			The gain  in EE increases as a function of the  SBS density. For a   normalized cache size of $0.4$ and an SBS density of $1.9\;\text{SBS}/Km^2$, we notice an increase of  $14.2\%$ in the achievable EE. This increase in EE  gain can be explained by the fact  that  the average transmit power is  a decreasing function of  SBS  density. }

		\begin{figure}[!htb]
			\centering	
			\includegraphics[scale=0.8]{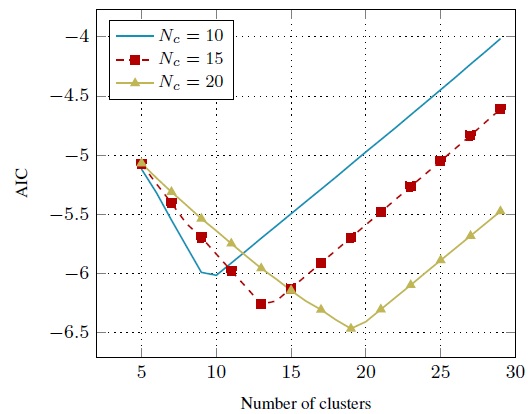}
			\captionsetup{font=footnotesize}
			\caption{Akaike information criterion}
		\end{figure}

		{	Figure 4 shows the performance of  AIC  model selection. We consider three settings in which, the true numbers of clusters are $10$, $15$ and $20$,  respectively. The  figure represents the computed AIC per point over the estimation range.  The lowest AIC value represents the model that  strikes the best trade-off between  fitness and complexity. Note that the negative values of the AIC are due to a negative bias which characterize the AIC with  a small sample number.}

		%	In Figure 6, we show the numerical results of energy efficiency  versus  the normalized cache size for different SBS densities.
		%	With a given cache size, the EE  has a quasi-concave behavior as a function of the SBS density.  We can see that  the SBS density has a huge impact on the EE. This is mainly due to  the infrastructure power consumption. Although increasing the SBS  density reduces the average transmit power, we see that this impact is minor compared with the energy expenses related to the  active infrastructure. 
		%	Increasing the normalized cache capacity results in an improvement in the EE efficiency. This means that increasing the cache size in each SBS is a very  practical way  to improve the system performances without the need for costly infrastructure.
		\begin{figure}[!htb]
			\centering	
			\includegraphics[scale=0.8]{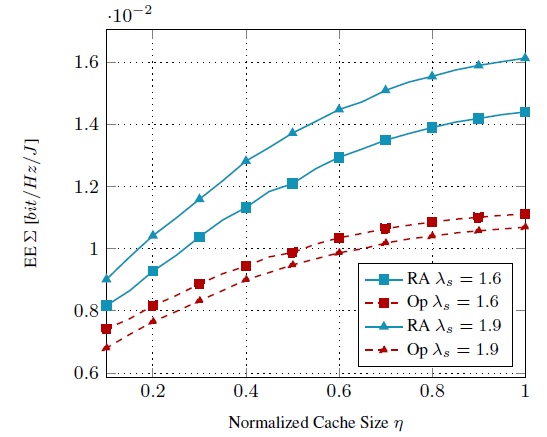}
			\captionsetup{font=footnotesize}
			\caption{EE vs normalized cache size $\eta$ with different SBS densities (RA= random SBS allocation, Op= optimized SBS allocation)}
		\end{figure}

		Figure 5 shows the impact of the SBS allocation algorithm on EE. We can see that, for different values of the  SBS density, optimizing the SBS allocation results in a considerable gain in the  EE. For a SBS density of $\lambda_s=1.9$ and a normalized cache size of $0.4$, optimizing the allocation of the SBSs results in an EE gain of $42.2\%$. As the  SBSs density increases, the  allocation algorithm  results in greater improvement in EE. Optimizing the cluster-SBS association  results in less average transmit power which reduces the interference and improves the achievable EE. 
		\section{Conclusion }
		In this paper, we studied a cache enabled small cell network. A content popularity based clustering approach was considered in order to  exploit the correlation between the request patterns of users.
		We develop an algorithm that enables to estimate the number of content popularity based clusters and to  efficiently assign users  to their respective groups.
		By considering the distribution of SBSs to be a Poisson point process, we investigated  the impact of proactive caching with user clustering on the achievable  EE .
		We proposed an optimization framework that enables to derive the  optimal active SBS  density vector  in order to  maximize the   EE. In order to exploit any  spatial correlation in user request patterns, we proposed a SBS allocation algorithm that aims at minimizing the transmit power by  bringing the  cached files closer to  the users that are  most likely to request them.
		Finally, we performed numerical analysis which show that the proposed clustering framework considerably
		outperforms the scheme in which the most popular files are cached in all SBSs. It also  shows that optimized SBS allocation results in an improvement in the achievable hit probability and EE.
		For  future  work,  the optimal cache placement strategy  with heterogeneous popularity profiles and mobility patterns will be investigated.
		\section{Appendix}
		\textbf{Appendix A proof of Lemma 1:}\\      
		{	We derive the expression of the average consumed power in the network with cache enabled SBSs. The average total power $\rho^c_{total}$ is given by:
			$$\rho^c_{total}= \mathbb E \left\lbrace \rho_I\right\rbrace + 
			\mathbb E \left\lbrace \rho_T\right\rbrace +\mathbb E \left\lbrace \rho_f\right\rbrace, $$
			where $\mathbb E \left\lbrace \rho_I\right\rbrace = \rho \lambda_s \pi R^2_n$ and       $\mathbb E \left\lbrace \rho_f\right\rbrace =\lambda_s \pi R^2_n\big(  \rho_{hd}  \mathbb P \left\lbrace hit\right\rbrace +  \rho_{bh} \big( 1-\mathbb P \left\lbrace hit\right\rbrace\big) \big).   $ 	 
			Taking into account the considered system model, the average transmit power used by a given user from  cluster $k$, $\mathbb E \left\lbrace \rho^{[k]}_T\right\rbrace$ can be written  as follows:
			\begin{align}\tag{31}
			\nonumber & \mathbb E \left\lbrace \rho^{[k]}_T\right\rbrace =     \mathbb E \left\lbrace \rho_k\right\rbrace (1- \sum_{j \neq k}^{} \sum_{i \in \Delta_s}^{} p_{iu}  ( 1- e^{- \lambda_{sj}\pi R^2}))\\
			\nonumber &	+ \sum_{j \neq k}^{} \sum_{i \in \Delta_s}^{} p_{iu}  ( 1- e^{- \lambda_{sj}\pi R^2})\mathbb E \left\lbrace \rho_j\right\rbrace . \nonumber
			\end{align}						
			After averaging over all users in the  network, the average consumed transmit power  is given by:
			\begin{align}\tag{32}
			\nonumber & \mathbb E \left\lbrace \rho_T\right\rbrace = \frac{\lambda_s \pi R^2_n}{U}  \sum_{k=1}^{N_c}  \sum_{u \in \Upsilon_k}^{}  \big(    \mathbb E \left\lbrace \rho_k\right\rbrace\\
			\nonumber & + \sum_{j \neq k}^{} \sum_{i \in \Delta_s}^{} p_{iu}  ( 1- e^{- \lambda_{sj}\pi R^2})(\mathbb E \left\lbrace \rho_j\right\rbrace  -\mathbb E \left\lbrace \rho_k\right\rbrace)\big). \nonumber
			\end{align}
			We need then to compute the average power used  by the users to communicate with the nearest SBS from any given cluster $k\;, k=1..N_c$.\\ 
			According to the PPP assumption for the location of the SBSs, the distance from a user to its  nearest SBS from cluster $k$, denoted by $r_k$, has the following \textit{pdf} \cite{power}:
			\begin{align}\tag{33}
			\nonumber&f_{r_k}(r)=2\pi \lambda_{sk} r e^{-\lambda_{sk} \pi r^2}. \nonumber
			\end{align}
			The transmit power used by the user in this case is given by $\rho_k=\rho_0 r^\alpha_k$. Then:
			\begin{align}\tag{34}
			\E{\rho_k} \nonumber &  =  \int_{0}^{R}   \;2\pi \lambda_{sk} r^{\alpha+1}  exp(-\lambda_{sk} \pi r^2)  \text{d} r\\
			\nonumber &  = \frac{\rho_0 \gamma(\frac{\alpha}{2}+1,\pi \lambda_{sk} R^2)}{(\lambda_{sk} \pi)^{\frac{\alpha}{2}} }.\nonumber
			\end{align}
			Following the same calculus for $\E{\rho_k}, \; k=1..N_c$, we obtain  the final expression of the average consumed power in the network:
			\begin{align}\tag{35}
			\nonumber & \rho^c_{total}=\lambda_s \pi R^2_n(  \rho_{hd}  \mathbb P \left\lbrace hit\right\rbrace +  \rho_{bh} \big( 1-\mathbb P \left\lbrace hit\right\rbrace\big)  +\rho)\\
			\nonumber &	+ \frac{\lambda_s \pi R^2_n}{U}  \sum_{k=1}^{N_c}  \sum_{u \in \Upsilon_k}^{}( \frac{\rho_0 \gamma(\frac{\alpha}{2}+1,\pi \lambda_{sk} R^2)}{(\lambda_{sk}\pi)^{\frac{\alpha}{2}} }\\
			\nonumber &   + \sum_{j \neq k}^{} \sum_{i \in \Delta_s}^{} p_{iu}  ( 1- e^{-  \lambda_{sj}\pi R^2})(\frac{\rho_0 \gamma(\frac{\alpha}{2}+1,\pi \lambda_{sj} R^2)}{(\lambda_{sj}\pi)^{\frac{\alpha}{2}} }  \\
			\nonumber & 	-\frac{\rho_0 \gamma(\frac{\alpha}{2}+1,\pi \lambda_{sk} R^2)}{(\lambda_{sk}\pi)^{\frac{\alpha}{2}} })). 
			\end{align}
			\textbf{Appendix B proof of Lemma 2:}\\
			We derive the achievable coverage probability when using channel inversion power control:
			\begin{align}
			\mathbb P\left\lbrace  \text{SINR} \geq \theta \right\rbrace \nonumber & =\E{  \mathbb P ( \left\|  h_u\right\| ^2  \geq (\frac{\sigma^2 + \sum_{k=1}^{N_c} I_{k}}{\rho_0}) \theta ) \arrowvert I_{k} \forall k}\\
			\nonumber & 	= \E{  exp (- \frac{\theta}{\rho_0} (\sigma^2 +\sum_{k=1}^{N_c} I_{k})) \arrowvert I_{k} \forall k }\\\tag{36}
			\nonumber & 	=   exp (- \frac{\theta}{\rho_0} \sigma^2) \prod_{k=1}^{N_c} \mathcal{L}_{I_{k}}(\frac{\theta}{\rho_0}).\nonumber
			\end{align}
			We use the fact that $\left\|  h_u\right\| ^2$ is exponentially distributed and $\mathcal{L}_{I_{k}}(s)$ is the Laplace transform of $I_{k}$ at $s$.
			To prove  Lemma 3,  we need to compute the Laplace transform of $I_{k}, \forall k $. 
			The interfering base stations constitute multiple PPP processes $\phi_{sk}, k=1...N_c$, each associated with a given cluster.
			The Laplace transform  of $I_{k}$  for a given $k$ is obtained as:
			\begin{align}
			\mathcal{L}_{I_{k}}(\frac{\theta}{\rho_0}) \nonumber &  = \E{exp \big(- \sum_{i \in \phi_{sk}}^{}  	\rho_ {ik}  \left\| h_{ui}\right\| ^2 r^{-\alpha}_{ui}      \big) } \\
			\nonumber & =exp\big(-2\pi \lambda_{sk}  \int_{0}^{\infty }\big( 1- \E{e^{-\theta \left\| h\right\| ^2 \rho_{k} r^{-\alpha}}} \big) r \; \text{d}\; r \big) \\\nonumber\tag{37}
			\nonumber &  = exp\big(-\pi \lambda_{sk}  (\frac{\theta}{\rho_0})^{\frac{2}{\alpha}} \E{\rho^{\frac{2}{\alpha}}_k}  \Gamma(1+\frac{2}{\alpha}) \Gamma(1-\frac{2}{\alpha}) \big).\nonumber
			\end{align}
			$ \E{\rho^{\frac{2}{\alpha}}_k}$  depends on the density of the small cells caching files from cluster $k$. $ \E{\rho^{\frac{2}{\alpha}}_k}$ can be deduced from \textbf{Appendix A} as; $\E{\rho^{\frac{2}{\alpha}}_k} = \frac{\rho_0^{\frac{2}{\alpha}} \gamma(2,\pi \lambda_{sk} R^2)    }{\lambda_{sk} \pi }$.
			Based on Slivnyak's {Theorem} for Poisson Point Processes \cite{power}, the obtained Expression is valid for any user within the network.\\
			\textbf{Appendix C proof of {Theorem }1:}\\
			We  start by showing that the objective function is quasi-concave. Given the expression of $\Sigma$  as a function of the density vector $\Lambda_s$, it is difficult to prove its quasi-concavity by using its gradient or Hessian matrix.
			However, using the fact that   a composition with an affine function preserves quasi-concavity \cite{boyd}, this proof can be  considerably simplified.  
			To prove the  quasi-concavity of  $\Sigma$ as a function of the density vector $\Lambda_s$, we consider an  affine function $f(t)$ given by:
			\begin{align}\tag{38}
			\nonumber &  f(t)=t Z +\Lambda^0_s,
			\end{align}
			where $\Lambda^0_s \in \mathbb{R}^{N_c \times 1}$ such that $\sum_{k=1}^{N_c} \lambda^0_{sk} \leq \lambda_{s_{max}}$, $Z \in \mathbb{R}^{N_c \times 1}$ and $t \in \mathbb{R}$.
			Since a composition with an affine function preserves quasi-concavity, it is sufficient to prove the  quasi-concavity of   $\Sigma(t Z +\Lambda^0_s)$ with respect to $t$ in order to show the quasi-concavity of   $\Sigma$  with respect to  $\Lambda_s$.\\  							
			The  objective  $\Sigma(t Z +\Lambda^0_s)$ can be written as  a product of two nonnegative  functions:\\
			$U(t Z +\Lambda^0_s)=\sum_{k=1}^{N_c}(t  z_k +\lambda^0_{sk}) \pi R^2_n \; \text{log} (1+\theta) \; \mathbb P \left\lbrace \text{SINR}  \geq \theta \right\rbrace$ and $V(t Z +\Lambda^0_s)= \frac{1}{\rho^c_{total} }$.\\
			We start  by computing the derivatives of   $U(t Z +\Lambda^0_s)$ and  $V(t Z +\Lambda^0_s)$ with respect to $t$:
			\begin{align}\tag{39}
			\nonumber &  U'=\frac{\rm d U(t Z +\Lambda^0_s)}{\rm d t}= \\
			\nonumber & (\sum_{k=1}^{N_c} z_k )  \pi R^2_n \; \text{log} (1+\theta) exp (- \frac{\theta}{\rho_0} \sigma^2) \prod_{k=1}^{N_c} \mathcal{L}_{I_{ku}}(\frac{\theta}{\rho_0})\\
			\nonumber &  \times \big(1- \sum_{k=1}^{N_c} \Gamma(1+\frac{2}{\alpha}) \Gamma(1-\frac{2}{\alpha}) \theta^{\frac{2}{\alpha}} \pi z_{k} R^2 e^{-  (tz_{k}+\lambda^0_{sk})\pi R^2}\big).\nonumber
			\end{align}
			Then $\frac{\rm d U(t Z +\Lambda^0_s)}{\rm d t} >0$. We do the same to $V(t Z +\Lambda^0_s)= \frac{1}{\rho^c_{total} }$. We have:
			\begin{align}\tag{40}
			V'\nonumber & =\frac{\rm d V(t Z +\Lambda^0_s)}{\rm d t}  = \frac{-\chi}{P^{c^2}_{total}},\nonumber
			\end{align}
			where $\chi$ is given by:
			\begin{align}\tag{41}
			\chi \nonumber &  =  (\sum_{k=1}^{N_c} z_k ) \pi R^2_n ((\rho_{hd}-\rho_{bh}) \mathbb P \left\lbrace hit\right\rbrace+ \rho+\rho_{bh})\\
			\nonumber & 	+\frac{\pi R^2_n}{U} ((\sum_{k=1}^{N_c} z_k )\chi_1 + (\sum_{k=1}^{N_c} tz_k +\lambda^0_{sk})  \chi_2) + \frac{\pi R^2_n}{U} \\
			\nonumber &  (\sum_{k=1}^{N_c} tz_k +\lambda^0_{sk})   \sum_{u=1}^{U} \sum_{k=1}^{N_c}  \big( \sum_{i \in \Delta_k}^{} p_{iu}\big)  z_{k} \pi R^2   e^{- (tz_{k}+\lambda^0_{sk}) \pi R^2}.\nonumber      
			\end{align}						
			Here $\chi_1$ and $\chi_2$ are respectively given by:
			\begin{align}\tag{42}
			\nonumber & \chi_1=  \sum_{k=1}^{N_c}  \sum_{u \in \Upsilon_k}^{}( \frac{\rho_0 \gamma(\frac{\alpha}{2}+1,\pi ( tz_k +\lambda^0_{sk}) R^2)}{(( tz_k +\lambda^0_{sk})\pi)^{\frac{\alpha}{2}} }\\
			\nonumber &  + \sum_{j \neq k}^{} \sum_{i \in \Delta_s}^{} p_{iu} \rho_0 ( 1- e^{- ( tz_j +\lambda^0_{sj})\pi R^2})\\
			\nonumber  & (\frac{ \gamma(\frac{\alpha}{2}+1,\pi ( tz_k +\lambda^0_{sj}) R^2)}{(( tz_j +\lambda^0_{sj})\pi)^{\frac{\alpha}{2}} }  -\frac{ \gamma(\frac{\alpha}{2}+1,\pi ( tz_k +\lambda^0_{sk}) R^2)}{(( tz_k +\lambda^0_{sk})\pi)^{\frac{\alpha}{2}} })),
			\end{align}	
			\begin{align}
			\nonumber &  \chi_2= \rho_0 \sum_{k=1}^{N_c}  \sum_{u \in \Upsilon_k}^{}(  (z_k \pi R^\alpha e^{- ( tz_k +\lambda^0_{sk})\pi R^2}\\
			\nonumber & 	-\frac{  \gamma(\frac{\alpha}{2}+1,\pi ( tz_k +\lambda^0_{sk}) R^2) \frac{\alpha z_k \pi}{2} (( tz_k +\lambda^0_{sk}) \pi)^{\frac{\alpha}{2}-1}    }{(( tz_k +\lambda^0_{sk})\pi)^{{\alpha}} }) \\
			\nonumber &	 + \sum_{j \neq k}^{} \sum_{i \in \Delta_s}^{} p_{iu}  ( 1- e^{- ( tz_j +\lambda^0_{sj})\pi R^2})	\\
			\nonumber & ( \frac{  \gamma(\frac{\alpha}{2}+1,\pi ( tz_k +\lambda^0_{sk}) R^2) \frac{\alpha z_k \pi}{2} (( tz_k +\lambda^0_{sk}) \pi)^{\frac{\alpha}{2}-1}    }{(( tz_k +\lambda^0_{sk})\pi)^{{\alpha}} }\\
			\nonumber &   -  z_k R^\alpha e^{- ( tz_k +\lambda^0_{sk})\pi R^2}     +\pi( z_j R^\alpha e^{- ( tz_j +\lambda^0_{sj})\pi R^2} \\
			\nonumber &	-\frac{  \gamma(\frac{\alpha}{2}+1,\pi ( tz_j +\lambda^0_{sj}) R^2) {\alpha z_j } (( tz_j +\lambda^0_{sj}) \pi)^{\frac{\alpha}{2}-1}    }{2(( tz_j +\lambda^0_{sj})\pi)^{{\alpha}} })) \\				
			\nonumber & 	+\sum_{j \neq k}^{} \sum_{i \in \Delta_s}^{} p_{iu} z_j\pi R^2 e^{- ( tz_j +\lambda^0_{sj})\pi R^2}\\
			\nonumber &  (\frac{\gamma(\frac{\alpha}{2}+1,\pi ( tz_j +\lambda^0_{sj}) R^2)}{(( tz_j +\lambda^0_{sj})\pi)^{\frac{\alpha}{2}} }  -\frac{ \gamma(\frac{\alpha}{2}+1,\pi ( tz_k +\lambda^0_{sk}) R^2)}{(( tz_k +\lambda^0_{sk})\pi)^{\frac{\alpha}{2}} })). 
			\end{align}	
			Consequently $\frac{\rm d V(t Z +\Lambda^0_s)}{\rm d t} <0$. 
			In what follows we  distinguish two cases depending on the existence of a point $t^*$  such that $  \frac{\rm d \Sigma(t Z +\Lambda^0_s)}{\rm d t}_{t=t^*}  = 0$.
			If  $ \exists  t^*$ such that $  \frac{\rm d \Sigma(t Z +\Lambda^0_s)}{\rm d t}_{t=t^*}  = 0$ then:
			\begin{align}\tag{43}
			\nonumber &  \frac{\rm d \Sigma(t Z +\Lambda^0_s)}{\rm d t}_{t=t^*}  = 0  \Leftrightarrow U'V+V'U=0 \Leftrightarrow \frac{-U'V}{V'U}=1.\nonumber 
			\end{align}
			We compute the derivative of $L(t)=\frac{-U'V}{V'U}$ with respect to $t$. The expression of the derivative is omitted here for brevity. We find that  $L'(t)>0$. Since $L(t)$ is a strictly increasing function then, according the {Theorem} of intermediate value, if $\exists  t^* $ such that  $L(t^*)=1$  then this point is unique. 
			Finally, depending on the existence  of $t^*$, we  have two cases:
			\begin{itemize}
				\item	If $\exists \; t^* $ such that  $L(t^*)=1$  then this point is unique and $\Sigma(f(t))$ is increasing for $t<t^*$ and decreasing for $t>t^*$.
				\item	If, on the other hand, $  t^* $ does not exists, then  $\Sigma(f(t))$ is a strictly monotone function.								
			\end{itemize}
			This proves that  $\Sigma(f(t))$ is a quasi-concave function of $t$. Since composition with  an affine function preserves quasi-concavity, we can deduce that   $\Sigma(\Lambda_s)$  is a quasi-concave function of $\Lambda_s$ and that, if $\exists \Lambda^*_s$ such that $\nabla \Sigma(\Lambda^*_s)=0$ then this vector is unique .\\
			In the  second  step of the proof, we  need to show that  the  constraint $ C(\Lambda_s) = \rho^{c}_{total}-\rho^{nc}_{total}$   is also quasi concave.  This is done in a similar way as in the first step by considering  $ C(f(t))$. After computing the derivative of $C(f(t))$ with respect to $t$, we find that:$ \frac{\rm d C(f(t))}{\rm d t}<0$.
			Then the first constraint is  quasi concave. Using the same method, it is trivial to show that the second constraint is also quasi-concave.  
			In order to finish the  proof,  we need to show that  the optimal solution can be found with zero duality gap.
			This will be done using results on quasi-concave programming from \cite{hoven}.
			Since $\frac{\rm d (\rho^{c}_{total}-\rho^{nc}_{total})}{\rm d \lambda_{sk}} \neq 0, \forall k=1...N_c$ then, according to  the  Necessity {Theorem} in \cite{hoven}, any solution  of the optimization problem $(23)$ satisfies the  KKT conditions.
			We, now, distinguish between two case: 
			\begin{itemize}
				\item 	If $\exists \; \Lambda^*_s $ such that $\nabla \Sigma(\Lambda_s)_{\Lambda_s=\Lambda^*_s} =0 $ and  $\Lambda^*_s$ satisfies the constraints then,  $\Lambda^*_s$ is unique and it is a global optimum of $(23)$. The uniqueness of $\Lambda^*_s$, if it exists, was shown in the  first step of the proof.
				\item 	If  $ \nabla \Sigma(\Lambda_s) \geq 0 $, $ \forall \Lambda_s \; \text{such that}  \; \Lambda_s^{\dagger} \mathbf{1} \leq \lambda_{s_{max}} $, the sufficiency {Theorem} in \cite{hoven} is verified. Consequently, by  combining the necessity and sufficiency results, the optimal SBS density vector can be derived  using KKT.
			\end{itemize}}
			\textbf{Appendix D proof of Lemma 3:}\\
			First we need to prove that the objective function $\Omega$ is sub-modular.	
			We consider two SBSs allocations $X \; \text{and} \; Y$ such that  $X \subseteq Y$  and we need to prove that the marginal value of adding a new allocated SBS $l$ to cluster $i$  in $X$ and $Y$ verifies:
			\begin{align}\tag{44}
			\nonumber& \Omega\big(X \cup \left\lbrace y_{li}\right\rbrace \big) - \Omega\big(X  \big) \geq \Omega\big(Y \cup \left\lbrace y_{li}\right\rbrace \big) - \Omega\big(Y  \big).\nonumber
			\end{align}	
			Monotonicity is trivial since any new SBS allocation cannot decrease the value of the objective function. In order to show submodularity of the function, we  compare the marginal values of adding $y_{li}$ to $X$ and  $Y$. Here we consider $\Pi_i(X \cup \left\lbrace y_{li}\right\rbrace)$ referring to the  users  that change their serving SBS from cluster  $i$. $\mu(u,i)$ refers to the index of  the SBS from cluster $i$ serving user $u$.     A user changes its serving SBS when the new allocated one is closer which induces less transmit power. Consequently, the marginal values of adding $y_{li}$ to $X$ and $Y$ are as follows:
			\begin{align}\tag{45}
			\nonumber &	\Omega\big(X \cup \left\lbrace y_{li}\right\rbrace \big) - \Omega\big(X  \big) = \sum_{u \in \Pi_i(X \cup \left\lbrace y_{li}\right\rbrace)}^{}  \big(\sum_{f \in \Delta_i}^{}  p_{fu}    \big)\\
			\nonumber & \times \big(  \omega^{(X \cup \left\lbrace y_{li}\right\rbrace)}_{u\mu(u,i)} - \omega^{(X )}_{u\mu(u,i)} \big),\\\nonumber
			\nonumber &	\Omega\big(Y \cup \left\lbrace y_{li}\right\rbrace \big) - \Omega\big(Y  \big) = \sum_{u \in \Pi_i(Y \cup \left\lbrace y_{li}\right\rbrace)}^{}  \big(\sum_{f \in \Delta_i}^{}  p_{fu}    \big)\\
			\nonumber &	\times \big(  \omega^{(Y \cup \left\lbrace y_{li}\right\rbrace)}_{u\mu(u,i)} - \omega^{(Y )}_{u\mu(u,i)} \big).\nonumber
			\end{align}
			Since $X \subseteq Y$  we can deduce that $\Pi_i(Y \cup \left\lbrace y_{li}\right\rbrace) \subseteq \Pi_i(X \cup \left\lbrace y_{li}\right\rbrace)$. Since a user changes its serving SBS only when a closer allocated one is available then  $\omega^{(Y \cup \left\lbrace y_{li}\right\rbrace)}_{u\mu(u,i)} - \omega^{(Y )}_{u\mu(u,i)} > 0$ which proves that $  \Omega\big(X \cup \left\lbrace y_{li}\right\rbrace \big) - \Omega\big(X  \big) \geq \Omega\big(Y \cup \left\lbrace y_{li}\right\rbrace \big) - \Omega\big(Y  \big)  $.
			Consequently, $\Omega$ is a sub-modular set function.
			It is simple to verify that the constraints  $\sum_{s=1}^{N_s} y_{sk} \leq N_{sk}  , \forall k= 1..N_c$ are equivalent to a matroid constraints \cite{femto}. 	Then the considered optimization problem is equivalent to maximizing a sub-modular function subject to matroid constraints.

		\end{document}